\documentclass[12pt]{article}
\usepackage[utf8]{luainputenc}
\usepackage{geometry}
\usepackage{array}
\usepackage{rotating}
\usepackage{float}
\usepackage{multirow}
\usepackage{amsmath}
\usepackage{amssymb}
\usepackage{graphicx}
\usepackage{setspace}
\usepackage{esint}
\usepackage{chngcntr}
\usepackage{enumerate}
\usepackage{graphicx}
\usepackage[authoryear]{natbib}
\usepackage{color}
\usepackage{changes}
\usepackage{hyperref}
\usepackage{cleveref}

\usepackage{amsfonts}
\usepackage{amsthm}
\usepackage{fullpage}
\usepackage{slashbox}
\usepackage{multirow}
\usepackage{enumerate}
\restylefloat{table}
\usepackage{geometry}
\counterwithin{figure}{section}
\theoremstyle{plain}
\theoremstyle{definition}
\newtheorem{proposition}{Proposition}

\newtheorem{corollary}{Corollary}

\usepackage{color}
\definecolor{Blue}{rgb}{0,0,1}

\definecolor{Orange}{rgb}{1,0.5,0}

\definecolor{Green}{rgb}{0,1,0}

\title{Cheating in Ranking Systems}
\author{Lihi Dery\thanks{ Ariel University, Israel, lihid@ariel.ac.il} \and Dror Hermel \thanks{ Ryerson University, Canada, and Ariel University, Israel, drorhe@ariel.ac.il} \and Artyom Jelnov\thanks{ Ariel University, Israel, artyomj@ariel.ac.il}}

\begin{document}
\maketitle

\begin{abstract}
Consider an application (e.g., an \textit{app}) sold on an on-line platform (e.g., \textit{Google Play}), with the app paying a commission fee and, henceforth, offered for sale on the platform. The ability to sell the application depends on its customer ranking. Therefore, developers may have an incentive to promote their application’s ranking in a dishonest manner. One way to do this is by faking positive customer reviews. However, the platform is able to detect dishonest behavior (cheating) with some probability, and then proceeds to decide whether to ban the application. We provide an analysis, and find the equilibrium behaviors of both the application’s developers (cheat or not) and the platform (setting of the commission fee). We provide initial insights into how the platform’s detection accuracy affects the incentives of the app developers.
\end{abstract}

\textbf{Acknowledgment.} We wish to thank Christopher Thomas Ryan, Yair Tauman, Richard Zeckhauser and the anonymous reviewers for their helpful suggestions.

\section{Introduction}
%A recent blog post by Johnny Lin\footnote{https://medium.com/@johnnylin/how-to-make-80-000-per-month-on-the-apple-app-store-bdb943862e88} provides anecdotal evidence into the possible benefits of cheating users of the Apple App Store. In his post Lin discusses the ability to mislead consumers, resulting, in the short-run, in enhanced revenues which translate into a higher ranking in the App Store - a measure many users look for when searching for a new app to download. This cycle, where cheating results in an app's high ranking, making it more visible for new consumers, who usually download apps that are ranked higher rather than lower, is broken when Apple, using some undisclosed mechanism, removes the cheating app.
%The negative impact generated in the aforementioned scenario is multifold. From the end user side, a monetary investment - or rather loss - is made to purchase an app that does not generate much positive utility. From the developer side as a collective - those who cheat, gain profits in the short run - assuming most develop what they believe to be quality, utility generating apps, there is a both a short term monetary loss (as the apps they develop rank lower) and reputation loss, as users may be overly cautious before downloading new apps. This last issue also concerns the platform (e.g., the Apple App store) as it is associated with lost revenues.

Various systems allow users to rate items.  Using these ratings, the systems are then able to present a ranked list of items. Strategic agents may attempt to manipulate these ranked recommendations in order to increase their personal utility. However, these manipulations are costly. Furthermore, such manipulation attempts can be identified by inspection, which are also costly.

Consider, for example, an application (an app) and an on-line platform (e.g., the Apple App store) . The app may buy fake ratings which translate into a higher ranking on the App Store – a measure many users look for when searching for a new app to download. The negative impact generated in this scenario is multifold. From the end user side, a monetary investment - or rather loss - is made to purchase an app that does not generate much positive utility. Collectively, from the developer side, those who cheat gain profits in the short run, while for the others there is a short term monetary loss (as the apps they develop are ranked lower).

For the platform there is reputation loss (which can be associated with lost revenues), as users may be overly cautious before downloading new apps. Therefore, it is customary for the platform to use some mechanism to detect and remove cheating apps.

In this paper, we develop a model that studies the interaction between a platform and an application. The platform collects a fee from applications that want to use it. We study two cases: when the fee is exogenously provided and when the platform sets the fee. We show that from the application’s side, for an imperfect cheating detection technology, cheating will take place. We analyze how the quality of the detection algorithm and app rankings affect the incentive to cheat. Furthermore, we analyze the platform’s decision as to what commission fee to impose.

To the best of our knowledge, this is the first attempt to apply game theory into ranking systems in practice.

We begin with some background (section \ref{sec:related}). We then present our theoretical model (section \ref{sec:model}) and conclude with a discussion of our main findings and future research directions (section \ref{sec:discussion}). 

\section{Background}\label{sec:related}

We begin by surveying related work (section \ref{rw}), proceed to provide some intuition as to why rankings are significant enough so that people are willing to manipulate ratings in order to receive a high ranking (section \ref{rankings_matter}), and then survey methods for manipulation detection (section \ref{md}). 

\subsection{Related Works}\label{rw}
We suggest a new focal point to address cheating in ranking systems - an approach related to the well-established inspection games literature (cf. \citealp{avenhaus2002inspection} for a survey). The substantial difference that we implement is that while in inspection games one of the players decides whether to commit some violation, and another player, namely, the inspector, decides whether to perform a costly test to detect this violation, in our model a noisy alert of the violation is sent automatically. 

The notion of an automatically sent signal based on the action of one of the players appears in the literature in different contexts: industrial espionage (\citealp{barrachina2014entry}), international conflicts (\citealp{jelnov2017attacking}), and sports (\citealp{berentsen2002economics,kirstein2014doping}) to name a few.
In our case, we identify the favorable and adverse effects resulting in case the platform attempts to deter cheaters using an imperfect detection mechanism. 

Our paper is related to the economic law enforcement literature, which goes back to \citet{becker1968crime} and is surveyed in \citet{polinsky2007theory}. In our setting, we have an enforcer (the platform) and a potential violator (the application). We study a specific kind of violation: cheating in reviewer ratings. In our setting, this violation depends on the initial rating.  Moreover, the enforcer and the potential violator may have a common interest, because the application pays a commission fee to the platform. 

The work of \citet{darby1973free} resembles our topic in the sense that in their paper    a violation is wrong information given by a service supplier to a consumer. They study the existence of this kind of violation in a free market, and discuss how government intervention can reduce it. \citet{darby1973free} do not model strategic behavior by the government. In our paper, a platform, not a government, enforces honest behavior on an application, and we incorporate strategic considerations of the platform. 

The literature on economics of tort law, which can be traced to \citet{landes1984tort}, relates to our paper as the application may cause  damage to the platform. However, the tort law literature discusses how to cause one party to take care and prevent accidents which damage another party. Compensation for damage is the most common tool in tort law. In our case, no compensation is paid to the platform.

\subsection{Why Rankings Matter}\label{rankings_matter}
It has been shown that a website’s rank, not just its relevance, strongly and significantly affects the likelihood of a click (\citealp{glick2014does}).

As of March 2017, Google had 2.8 million apps available through its Android platform, and Apple had 2.2 million apps available on its iTunes App Store.\footnote{\textit{Statista: The Statistics} Portal https://www.statista.com/statistics/276623/number-of-apps-available-in-leading-app-stores/} With such massive numbers, users interested in discovering apps rely on rated listings, known as “top charts” such as “top free games”, “top free apps”, etc. Furthermore, a study by \citet{carare2012impact} indicates that users are willing to pay $\$4.50$ or more for an app that is top ranked as compared to the same app that is unranked, as people in general tend to disproportionately select products that are ranked at the top (\citealp{smith2001consumer,cabral2016box}). We further emphasize the monetary effects of higher ranking of apps by referring the reader to \citet{lee2014determinants}, who claim that one of the keys to a successful app is top rank status. 

Reviews play a critical role in online commerce (\citealp{mauri2013web}). For example, hotel reviews on websites that customers perceive as credible influence purchase behaviors (see \citealp{casalo2015online}). 
\citet{mayzlin2014promotional} show that competing products can self-promote by faking positive reviews for themselves, or negative reviews for competitors. 
For analysis of the impact of reputation in e-commerce, see also \citet{resnick2002trust}, who show that positive reviews of previous online transactions can predict good transactions in the future. Thus, the monetary gain from a highly ranked app is an incentive for app developers to boost their app rankings on the charts. In competing over reputation and higher ranking, product managers might be tempted to engage in manipulative behavior (\citealp{gossling2016manager}).

Unfortunately, in the app development context, some app developers choose to do so in a deceptive manner by paying fraudulent ranking services. If not addressed, these deceptions are harmful to the app platform – for developers, users and platform owners. For the developers, fraudulent ranking leads to unfair competition, which might discourage honest developers. The users might be led to install misbehaving or malicious apps, or they might be dissatisfied with platform app ratings and stop trusting the app charts. Finally, the platform’s reputation might be compromised.

\subsection{Manipulation Detection}\label{md}
Detection of deceptions is a top-priority for platforms, and is performed by detecting suspicious app patterns, user patterns, or both. A closely related line of work focuses on malware detection (\citealp{burguera2011crowdroid,narudin2016evaluation,seneviratne2017spam}). However, our focus is on ranking fraud, where the manipulating app is not necessarily malware.

Manipulating apps exhibit a different review pattern when compared with honest apps. Some algorithms for fake review detection focus on textual analysis of fake reviews, finding certain language constructs that are often used in fake reviews (e.g., \citealp{ott2011finding,hu2012manipulation,banerjee2017don}), while others (\citealp{schuckert2016insights}) point out the contradictions between overall (e.g. a hotel) and detailed rating of the same product (e.g.  specifics such as cleaning or location), which can expose fraudulent ratings.

Manipulated app rankings are likely to generate drastic ranking increases or decreases in a short time, or show strong rating deviations (see \citealp{zhu2013ranking}). \citet{heydari2016detection} show that the time interval in which the rating is given is also a measure of review trustworthiness. The detection can be based solely on ratings, by analyzing rating shifts, under the assumption that most ratings are honest (\citealp{akoglu2013opinion,savage2015detection}). In addition, manipulative users have a different review pattern. As the cost of setting up a valid account is quite high (e.g., to rate an app on Google Play requires a Google account, a mobile phone registered with the account, and the installed app), manipulators reuse their account and rate many apps over a short time frame. Manipulators may rate up to $500$ apps a day, rating them all with $5$ stars.

Manipulative users are usually part of a well-organized crowdsourcing system that performs malicious tasks. These are nicknamed “crowdturfing” systems, and their unique features have been mapped (see \citealp{wang2012serf}).
A recent study by \citet{Chen:2017} focused on identifying app clusters that are co-promoted by collusive attackers. The identification is based on unusual changes in rating patterns, measuring feature similarity in apps and applying machine learning techniques.
\citet{ye2015discovering} show that network information can also be employed. The baseline assumption is that an honest set of reviews for a product (or app) is formed by independent reviewer actions with various levels of activities and reviews. Therefore, a non-manipulative app should have reviews with various levels of network centrality. Furthermore, correlated review activities can be combined with linguistic and behavioral review signals (\citealp{rahman2017search}). %This is proposed be the FairPlay system

In the next section, we present a formal model for the interaction of two agents, an application and a platform. Our model assumes that one (or all) of the above manipulation detection capabilities are available to the platform.

\section{Model}\label{sec:model}

We consider two models: one with an exogenous fee and one with an endogenous fee i.e., a fee set by the platform. 

\subsection{Preliminaries}

We consider two risk-neutral agents: an application ($A$) and a platform ($P$). 
We focus only on cases when the application decides to enter the platform. 

A rating for each application is calculated periodically. The rating represents the opinion of the users and is observed both by the application and the platform. Upon entrance, at stage $t_0$ each application receives an initial rating of $r_0 = 0$. At stages $t_1, t_2$, the application obtains a rating of $r_1, r_2 \in [0,1]$, respectively. For simplicity, we denote $r_1$ as $r$.
We hereby study the last stage ($t_2$).

Naturally, for the application, a higher rating results in higher visibility on the platform translating into more profits. 
We assume the profit is proportional to the application's current rating $r_2$, minus a commission fee $f$  ($f \geq 0$) payed to the platform. 
Thus the application is left with a revenue of $\gamma r_2(1-f)$, where $\gamma >0$ is the proportion coefficient. 

In order to increase the rating $r_2$, the application may decide to cheat ($c$) (e.g., by adding fake ratings). 
If the application does not cheat ($\hat{c}$) it still has a probability  $l(r)$ to obtain the highest rating $r_2 = 1$. However, with probability $1-l(r)$ the rating is $r_2 = r$.
The probability $l(r)$ ($0 < l(r) < 1$) increases in $r$, namely, the higher the rating $r$, the more probable it is that the application will reach $r_2 = 1$. 

The platform has some imperfect algorithm, that enables it to detect applications that might be cheating (see Section \ref{md} for more details on such algorithms). 
Indeed, no algorithm or technology is 100\% error free, and the used algorithm might overlook some cheating applications as well as label honest applications as cheaters. 

At stage $t_2$,  a rating of $r_2 = 1$ triggers an automatic noisy alert sent to the platform. The alert $s$ means that the application is suspected of cheating ($\hat{s}$ means the opposite alert).
When the platform receives $s$ it is required to choose whether to ban ($b$) or not ban ($\hat{b}$) the application. The ban decision is equivalent to setting the application's rating to $r_2 = 0$. Note that this implies a different penalty cost for different applications; an application with a higher rating at stage $t_1$ has more to lose from a ban than an application with a lower rating. 

Let $\alpha(r)$ be a type-I probability error, namely, the probability that $s$ is sent when $A$ does not cheat, and let $\beta(r)$ be the probability of a type-II error, namely, the probability that $s$ is not sent when $A$ cheats. 

We consider $\alpha$ and $\beta$ as commonly known. 
When a platform chooses a cheating detection algorithm, as part of the acceptance testing performed when integrating it, the algorithm is tested on different scenarios (where the results are known), and $\alpha$ and $\beta$ can then be estimated. In a similar manner, application developers can estimate these parameters. 

We assume that $\beta(r)$ weakly increases in $r$. Namely, a high rating $r$ that increases to $r_2=1$, is less detectable than a low rating that increases to $r_2=1$.

The platform's utility consists of three factors:
the revenue from the commission fee that the application pays ($\gamma r_2f$),
the cost of non-detection (denoted $w$), and
the cost of false accusation (denoted $v$). 
The two latter costs can be interpreted as a loss of the platform's reputation, which translates into loss of user confidence in the platform, leading to a decrease in purchases and thus in the platform's revenues. 

Consequently, if the application cheats and is not banned, $P$'s utility is  $\gamma r_2f-w$, $w>0$. However, if the application does not cheat and is not banned, $P$ obtains $\gamma r_2f+v$, $v>0$.  If the application is banned, the platforms revenue is $0$. 
The game and player utilities are defined in Figure \ref{g_descr}. 

\begin{figure}[h]
\centering
\includegraphics[scale=0.5]{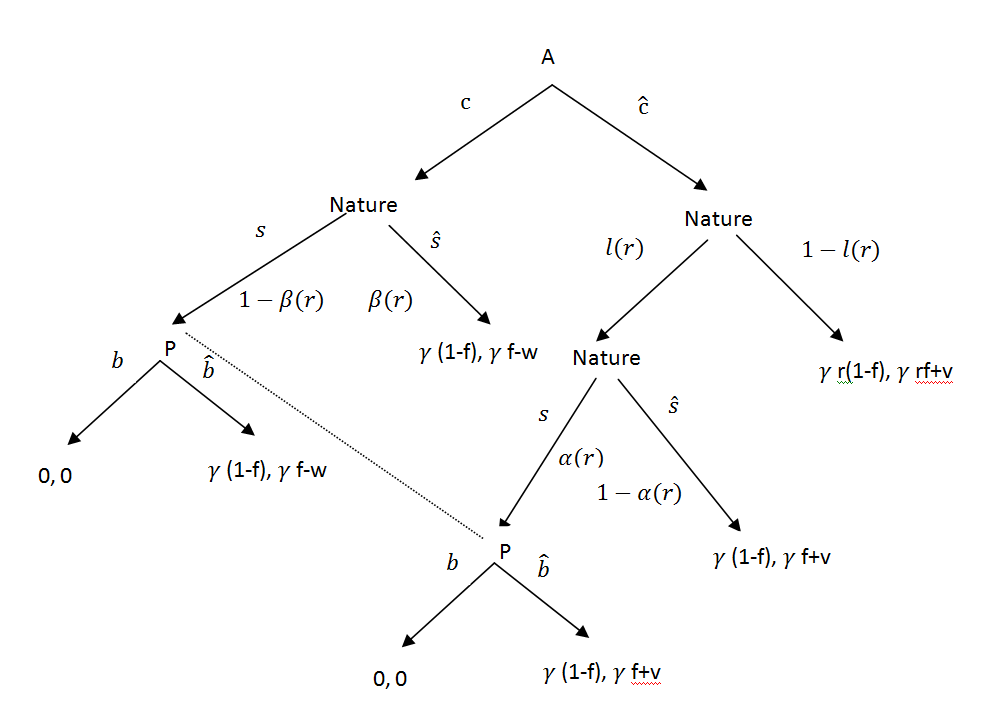}
\caption{Description of the game. Each pair of utilities represents application and platform utilities, respectively. }
\label{g_descr}
\end{figure}

If $r=1$, the highest rating is guaranteed to the application. Trivially, in this case there is an equilibrium where the application does not cheat, and the platform does not ban it. We assume hereafter $r<1$.
\subsection{Exogenous fee}\label{exogenous}
We now proceed and describe our results for an exogenous fee. 
The platform receives a signal $s$. If $\alpha(r) = 0$ (i.e., there is no possibility for a mistake regarding the signal), the platform bans the application. In the other cases: 
\begin{enumerate}
    \item If the platform's revenue from the commission fee is higher than the cost of non-detection (i.e. the reputation loss) then it will not ban the application even if it suspects cheating. Consequently, the application will surely cheat. 
    \item If $\beta(r)$ is high then the application is encouraged to cheat, since it is likely that cheating will not be detected. 
    \item If none of the former occurs, then the application will cheat with some positive probability, but not with certainty. 
\end{enumerate}

Formally, the following proposition characterizes an equilibrium of the game. 

\begin{proposition} \label{equilibrium}
Let $\alpha(r) > 0$. 
\begin{enumerate}
\item If $w<\gamma f$, then in the unique equilibrium of $G$ the application cheats with certainty, and with certainty, the platform does not ban the application.  
\item If $l(r)-\alpha(r)l(r)+r-rl(r)<\beta(r)$, then in the unique equilibrium of $G$ the application cheats with certainty, and following an alert $s$, the platform bans with certainty.
\item If $w>\gamma f$ and $l(r)-\alpha(r)l(r)+r-rl(r)>\beta(r)$, then the application  cheats with a probability $P_c$, $0<P_c<1$, and following an alert $s$ the platform bans with probability $P_b$, $0<P_b<1$. 
\end{enumerate}
\end{proposition}

All proofs appears in the appendix.

Note, that for  $w=\gamma f$ or $l(r)-\alpha(r)l(r)+r-rl(r)=\beta(r)$ the equilibrium may not be unique.

We now focus on part (3), where cheating occurs with some probability $0<P_c<1$. In this case, an increase in $\alpha(r)$, $\beta(r)$, $v$ and $f$ and a decrease in $w$ results in more cheating. 

The application compensates for an increase in the commission fee ($f$) by cheating, which may lead to a higher rating and thus an increase in profit. 
An increase in the false accusation cost ($v$) means that the platform is more reluctant to ban an application, which also leads to more cheating. 

An increase in the probability of false accusation $\alpha(r)$ means that even an honest application is more likely to be (mistakenly) banned for cheating, thus the application has less incentive to act honestly. Consequently, the probability the platform will ban applications increases in $\alpha(r)$.
An increase in the probability of non-detection of cheating $\beta(r)$ as well as a decrease in non-detection cost ($w$)  will lead to more cheating. Formally, and as a direct result of part 3 in Proposition~\ref{equilibrium}:

\begin{corollary}\label{behavior}
Let $w>\gamma f$ and $l(r)-\alpha(r)l(r)+r-rl(r)>\beta(r)$. Then in the equilibrium of $G$:
\begin{enumerate}
\item The probability that the application cheats increases in: $f$, $v$, $\alpha(r)$, and $\beta(r)$, and decreases in $w$.
\item The probability that $P$ bans $A$, following alert $s$, increases in $\beta(r)$ and in $\alpha(r)$. 
\end{enumerate}
\end{corollary}

When $\alpha(r)$ is independent of $r$, and cheating occurs with some probability $0<P_c<1$, then as the application approaches a rating of $1$, the more it is likely to cheat. The intuition behind this is that since applications with a rating  close to $1$ are less suspected of cheating, and hence less likely to be detected, they can thus cheat more freely.

\begin{corollary}\label{increasesr}
Suppose $\alpha(r)$ is constant, $\alpha(r) \equiv \alpha$. Let $w>fa$ and $l(r)-\alpha(r)l(r)+r-rl(r)>\beta(r)$. Then in the equilibrium of $G$ the probability that $A$ cheats increases in $r$.
\end{corollary}

Corollary \ref{increasesr}, that cheating increases in $r$, is not very surprising for a strictly increasing $\beta(r)$ since this means that less alerts $s$ are sent when $r$ is close to $1$. However, what is somewhat astounding is the fact that this claim holds even for those cases where $\beta(r)$ is constant, meaning - that even when the probability of the alert is constant, cheating increases in $r$. 

\subsection{Endogenous fee }
In section \ref{exogenous} the fee $f$ is exogenous. However, in reality the fee is set by the platform. Consider a model where at an initial stage $t_0$ the platform chooses a fee $0\leq f\leq1$, and then the game proceeds as defined in Figure \ref{g_descr}.\footnote{The platform maximizes its expected utility. The technicalities are characterized in proposition \ref{EUP} in the Appendix.}

An increase in the commission fee affects the platform’s utility in two ways: (1) positive – an increase in revenue and (2) negative – an increase in cheating, which reduces the platform’s utility, a direct consequence of Corollary~\ref{behavior}. 
 
As Figure \ref{res1} illustrates, the platform may maximize its expected utility for a fee lower than $1$, and the platform collects a higher fee when the alert it obtains is more precise (a lower $\beta(r)$). 
%In figure \ref{res4}, the fee is higher when the rating $r$ is lower. 

These results are not general. As can be seen in Figure \ref{res5}, when the cost of non-detection is set to $w=4$ instead of $w=3$, the platform maximizes its expected utility for a fee of $f=1$.

\begin{figure}[H]
\centering
\includegraphics[scale=0.6]{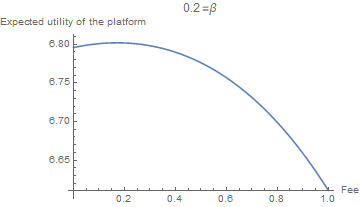}
\includegraphics[scale=0.6]{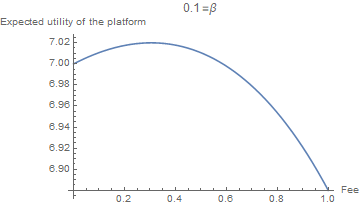}
\includegraphics[scale=0.6]{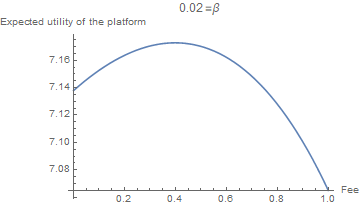}
\caption{Result for various $\beta(r)$ with $\gamma=1$, $r=0.6$, $\alpha(r)=0.1$, $l(r)=0.6$, $v=9$, $w=3$}
\label{res1}
\end{figure}

% \begin{figure}[H]
% \centering
% \includegraphics[scale=0.6]{res2.png}
% \includegraphics[scale=0.6]{res4.png}
% \caption{Result for two $r$, with $\gamma=1$, $\alpha(r)=0.1$, $\beta(r)=0.1$, $l(r)=0.6$, $v=9$, $w=3$}
% \label{res4}
% \end{figure}

\begin{figure}[H]
\centering
\includegraphics[scale=0.6]{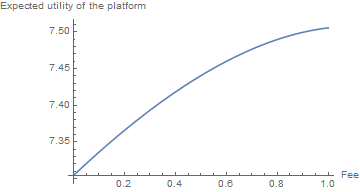}
\caption{Result for $w=4$ with $\gamma=1$, $r=0.6$, $\alpha(r)=0.1$, $\beta(r)=0.1$, $l(r)=0.6$, $v=9$}
\label{res5}
\end{figure}

\section{Conclusions}\label{sec:discussion}
In this paper, we provide a novel stylized framework to study the interaction between an app sales platform (e.g., Apple's app store or Google Play) and an app developer who may be tempted to cheat in order to increase their app ranking. Our framework captures some of this interesting interaction, and the consequential equilibrium analysis gives rise to some important implications.

Our most significant finding is that a higher fee leads to more cheating. Consequently, even a monopolistic platform may choose not to impose a high fee. 

%\todo{Applications with a higher rating are thus more probe to cheating. It follows that the platform should focus its efforts on these applications (Artyom, you said given alpha and beta are known. but must we really repeat that? alpha and beta have to be known}

Furthermore, we found that precise alert signals decrease cheating; when the cheating detection algorithm is a good one (i.e. $\alpha$ and $\beta$ are low) less cheating occurs. 
Thus, we conclude that if the platform has a good manipulation detection algorithm then it should make this known (i.e. publicize its $\alpha$ and $\beta$), since the application developers will refrain from cheating if they know that there is a high chance they will be caught.  

Numerically,  we show that a lower non-detection rate allows the platform to request a higher fee. This result gives the platform yet another reason to invest in acquiring or developing a good manipulation detection algorithm.

%We have considered only cases where the product has an entrance cost. \todo{?}
%As an example of an entrance cost-free product, consider the fake news phenomenon. Platforms such as honorable TV stations or online sites might try to verify information before they publish it, in order to retain their public reputation. However, a delay in publication might mean that some other channel will publish first and win the public's interest. 
%From the abundance of fake news, we hypothesize that either the platform prefers to risk its reputation over the risk of seeming "outdated", or that the product (in this case - the fake news source) has no cost of spreading the fake news, and is nonpunishable. 

We focused on the commission fee the product has to pay the platform and considered other costs, such as promotion costs and the cost of creating fake reviews as negligible. 
We assumed that the cost of a fake review is sufficiently low, and the reward from an undetected cheat is high, so there is incentive to cheat. 
Note that platforms have different tools to protect against manipulations and make it more difficult to create a fake review. For example, they can require reviewers to use the verification system CAPTCHA$^{\copyright}$, or to verify that a reviewer really consumes the product. Still, sophisticated manipulators can bypass these barriers. 

Utilities in our model are exogenous. We assume that the platform is interested in its reputation, and that cheating harms the platform's reputation. In a future extended model, platform competition can be considered, where more than one platform competes for customers and customers abandon a platform if they are dissatisfied with its cheating prevention level (for platform competition in a non-cheating environment see, for example, \citealp{halaburda2016role}).

Most importantly, we provide initial insights into how the platform detection accuracy affects the incentives of the app developers. Understanding these interactions, and the resulting equilibria, provide ample foundations to address future points of interest. For example, it can be used to better understand how to put in place mechanisms that align incentives and provides a benchmark framework for future empirical work.

Our findings and conclusion are relevant to other types of e-commerce as well, and can be of interest in any scenario that involves a product that faces a fee to be rated and ranked on an online platform. Other examples of such systems may include online vendor sites such  as Amazon, eBay, and hotel bookings sites. 

\bibliographystyle{apa}
%\bibliography{references}

\section*{Appendix}
\begin{proof}[Proof of proposition \ref{equilibrium}]
Consider first an equilibrium where $P$ chooses pure $\hat{b}$. Then $c$ is a superior action of $A$. By Figure \ref{g_descr}, the expected utility of $P$ in this case is $\gamma f-w$. If $P$ chooses $b$, its payoff is $0$, and the platform prefers $\hat{b}$  to $b$ for $w< \gamma f$.

Next, consider $P$, following $s$, chooses pure $b$. Observe, that for $\alpha(r)>0$, pure $\hat{c}$ is not an equilibrium in this case. By contrary, if $A$ chooses pure $\hat{c}$, with positive probability it obtains the rating $1$ and  a false signal $s$ is sent to $P$; therefore, $b$ is not the best reply of the platform.

If $A$ chooses $c$, and $P$, following $s$, chooses pure $b$, $A$'s expected utility is $\gamma (1-f)\beta(r)$. If $A$ does not cheat, and $P$, following $s$, chooses pure $b$, $A$'s expected utility is $\gamma (1-f)[r(1-l(r))+l(r)(1-\alpha(r))]$. Thus, $A$ prefers $c$ to $\hat{c}$ for $\beta(r)>r-rl(r)+l(r)-\alpha(r)l(r)$.

Let $A$ choosing $c$ with probability $P_c$. Given alert $s$, let $P(c|s)$ be belief of $P$  that $A$ cheats:
\begin{equation} \label{Pcs}
P(c|s)=\frac{P_c(1-\beta(r))}{P_c(1-\beta(r))+(1-P_c)l(r)\alpha(r)}.
\end{equation}

 Following rating $1$ of $A$ and alert $s$, $P$ is indifferent between $b$ and $\hat{b}$ for
\[\gamma f-wP(c|s)+v(1-P(c|s))=0, \]
and by \eqref{Pcs}, this is equivalent to
\begin{equation}\label{Pc}
P_c=\frac{\alpha(r)l(r)(\gamma f+v)}{\alpha(r)l(r)(\gamma f+v)+(1-\beta(r))(w-\gamma f)}.
\end{equation}
By \eqref{Pc}, $0<P_c<1$ for $w>\gamma f$.

Let $P_b$ be a probability with which $P$ bans the Application, following alert $s$. $A$ is indifferent between $c$ and $\hat{c}$ for
\[\gamma (1-f)[1-(1-\beta(r))P_b]=\gamma (1-f)[(1-l(r))r+l(r)(1-\alpha(r)P_b)],\]
namely,
\begin{equation}\label{Pb}
P_b=\frac{(1-r)(1-l(r))}{1-\beta(r)-\alpha(r)l(r)}.
\end{equation}
 By \eqref{Pb}, $0<P_b<1$ for $l(r)-\alpha(r)l(r)+r-rl(r)>\beta(r)$.
\end{proof}

\begin{proof}[Proof of Corollary \ref{behavior}]
Since conditions of part 3 of proposition \ref{equilibrium} hold, probabilities of cheating and of banning are given by \eqref{Pc} and \eqref{Pb}, respectively. Results follow directly by \eqref{Pc} and \eqref{Pb}.
\end{proof}

\begin{proof}[Proof of Corollary \ref{increasesr}]
Since conditions of part 3 of proposition \ref{equilibrium} hold, probabilities of cheating is given by \eqref{Pc}. The result follows directly by \eqref{Pc} and \and by $\frac{\partial l(r)}{\partial r}>0$ and $\frac{\partial \beta(r)}{\partial r}\geq 0$.
\end{proof}

\begin{proposition}\label{EUP}
\begin{enumerate}
\item If $w<\gamma f$, then in equilibrium the expected utility of $P$ is $\gamma f-w$.  
\item If $l(r)-\alpha(r)l(r)+r-rl(r)<\beta(r)$,  then in equilibrium the expected utility of $P$ is $\beta(r)(\gamma f-w)$.
\item If $w>\gamma f$ and $l(r)-\alpha(r)l(r)+r-rl(r)>\beta(r)$, then in equilibrium the expected utility of $P$ is $(1-P_c)[\gamma f[(1-l(r))r+l(r)]+v]+P_c(\gamma f-w)$, where $P_c$ is given by \eqref{Pc}. 
\end{enumerate}
\end{proposition}
\begin{proof}
Directly from Proposition \ref{equilibrium} and Figure \ref{g_descr}.
\end{proof}
\end{document}